\documentclass[times,5p]{elsarticle}

\usepackage{amssymb}
\usepackage{amsmath}
\usepackage{graphicx}
\usepackage{algpseudocode}
\usepackage{color}
\usepackage[utf8]{inputenc}
\usepackage[T1]{fontenc}
\usepackage{balance}

\newtheorem{thm}{Theorem}
\newtheorem{lem}[thm]{Lemma}
\newtheorem{cor}[thm]{Corollary}
\newtheorem{obs}[thm]{Observation}

\newtheorem{defi}[thm]{Definition}
\newtheorem{inv}[thm]{Invariant}

\newproof{proof}{Proof}

\newcommand{\N}{\mathbb{N}}
\newcommand{\Nz}{\mathbb{N}_0}

\begin{document}
\begin{frontmatter}

\title{An Efficient Data Structure for\\ Dynamic Two-Dimensional Reconfiguration\tnoteref{ccc}\tnoteref{arcs}}
\tnotetext[ccc]{This work was supported by the DFG Research Group
  FOR-1800, ``Controlling Concurrent Change'', under contract number
  FE407/17-1.}
\tnotetext[arcs]{A preliminary extended abstract of this
  paper appears in ARCS2016~\cite{ARCS}.}

\author[tubs]{S\'andor P.\ Fekete\corref{cor}}\ead{s.fekete@tu-bs.de}
\author[tubs]{Jan-Marc Reinhardt}\ead{j-m.reinhardt@tu-bs.de}
\author[tubs]{Christian Scheffer}\ead{scheffer@ibr.cs.tu-bs.de}

\address[tubs]{Department of Computer Science, TU Braunschweig,
  Germany.}
\cortext[cor]{Corresponding author}

\begin{abstract}
In the presence of dynamic insertions and deletions into a partially reconfigurable FPGA, 
fragmentation is unavoidable. This poses the challenge of developing efficient approaches to 
dynamic defragmentation and reallocation.
One key aspect is to develop efficient
algorithms and data structures that exploit the two-dimensional geometry 
of a chip, instead of just one.
We propose a new method for this task,
based on the fractal structure of a quadtree,
which allows dynamic segmentation of the chip area,
along with dynamically adjusting the necessary communication
infrastructure. We describe a number of algorithmic
aspects, and present different solutions.
{\color{black} We also provide a number of basic simulations that indicate that the theoretical worst-case bound may be
pessimistic.}
\end{abstract}


\end{frontmatter}

\section{Introduction}
\label{sec:intro}

In recent years, a wide range of methodological developments on FPGAs
{\color{black} aim at combining} the performance of an ASIC implementation 
with the flexibility of software realizations. One important development
is partial runtime reconfiguration, which allows overcoming significant area overhead,
monetary cost, higher power consumption, or speed penalties (see e.g.~\cite{rose_FPGAgap}).
As described in~\cite{fks-ddrd-12}, the idea is to load a sequence of different 
modules by partial runtime reconfiguration. 

In a general setting, we are faced with a dynamically changing set of modules,
which may be modified by deletions and insertions. Typically, there is no full
a-priori knowledge of the arrival or departure of modules, i.e., we have to deal 
with an online situation. The challenge is to ensure that
arriving modules can be allocated. Because previously deleted modules may
have been located in different areas of the layout, free space may be fragmented,
making it necessary to {\em relocate} existing modules in order to provide
sufficient area. In principle, this can be achieved by completely {\em defragmenting} 
the layout when necessary; however, the lack of control over
the module sequence makes it hard to avoid frequent full defragmentation,
resulting in expensive operations for insertions if a na\"ive approach is used.

Dynamic insertion and deletion are classic problems of Computer Science.
Many data structures (from simple
to sophisticated) have been studied
that result in low-cost operations and efficient maintenance of
a changing set of objects. These data structures are mostly
one-dimensional (or even dimensionless) by nature, making it hard to 
fully exploit the 2D nature of an FPGA. In this
paper, we propose a 2D data structure based on a quadtree
for maintaining the module layout under partial reconfiguration and reallocation.
The key idea is to control the overall structure of the layout, 
such that future insertions can be performed with a limited
amount of relocation, even when free space is limited.

Our main contribution is to introduce a 2D
approach that is able to achieve provable constant-factor efficiency
for different types of relocation cost. To this end,
we give detailed mathematical proofs for a slightly simplified setting, along
with sketches of extensions to the more general cases. {\color{black} We also provide
basic simulation runs for various scenarios, indicating the quality of
our approach.}

The rest of this paper is organized as follows. The following Section~2 
provides a survey of related work. For better accessibility of the key 
ideas and due to limited space, our technical description 
in Section~3, Section~4, and Section~5  focuses on the case of discretized quadratic modules
on a quadratic chip area. We discuss in Section~6 how general
rectangles can be dealt with, with corresponding {\color{black} simulations}
in Section~7. {\color{black}Along the same lines, we do not explicitly elaborate on 
the dynamic maintenance of the communication infrastructure; see Figure~\ref{fig:config} for the basic
idea. Further details are left to future work, with groundwork laid in~\cite{meyer}.}

\section{Related Work}

The problem considered in our paper has a resemblance to one-dimensional
{\em dynamic storage allocation}, in which a sequence of storage requests of varying
size have to be assigned to a block of memory cells, such
that the length of each block corresponds to the size of the request.
In its classic form (without virtual memory), this block needs to be contiguous;
in our setting, contiguity of two-dimensional allocation is a must, as reconfigurable devices 
do not provide techniques such as paging and virtual memory.
Once the allocation has been performed,
it is static in space: after a block has been occupied,
it will remain fixed until the corresponding data is no longer needed
and the block is released. As a consequence, a sequence of
allocations and releases can result in fragmentation of
the memory array, making it hard or even impossible to store
new data. 

On the practical side,
classic buddy systems partition the one-dimensional storage into a number of standard block
sizes and allocate a block in a smallest free standard interval 
to contain it. Differing only in the
choice of the standard size, various systems have been 
proposed \cite{Bromley80,Hinds75,Hirs73,Know65,Shen74}.
Newer approaches based on cache-oblivious structures
in memory hierarchies
include Bender et al.~\cite{Bender05,Bender05a}.
Theoretical work on one-dimensional contiguous allocation
includes
Bender and Hu~\cite{bender_adaptive_2007}, who consider
maintaining $n$ elements in sorted
order, with not more than $O(n)$ space. 
Bender et
al.~\cite{bender_maintaining_2009} aim at reducing
fragmentation when maintaining $n$ objects that require
contiguous space.  Fekete et
al.~\cite{fks-ddrd-12} study
complexity results and consider practical applications on FPGAs.
Reallocations have also been studied in the context of heap
allocation. Bendersky and Petrank~\cite{bendersky_space_2012} observe
that full compaction, i.e., creating a contiguous block of free space
on the heap,
is prohibitively expensive and consider partial compaction. 
Cohen and Petrank~\cite{cohen_limitations_2013} extend these
to practical applications.
Bender et al.~\cite{bender_cost-oblivious_2014} 
describe a strategy that achieves good amortized movement costs
for reallocations, where allocated blocks {\color{blue}can} be moved at a cost to a new position that is disjoint
from with the old position. 
Another paper by the same authors~\cite{bender_reallocation_2014} deals
with reallocations in the context of scheduling. 
Examples for packing problems in applied computer science come from
allocating FPGAs. Fekete et al.~\cite{fekete_efficient_2014} examined
a problem dealing with the allocation of different types of resources
on an FPGA that had to satisfy additional properties. For example, to
achieve specified clock frequencies diameter restrictions had to be
obeyed by the packing. The authors were able to solve the problem
using integer linear programming techniques.

Over the years, a large variety of methods and results for
allocating storage have been proposed. The classical sequential fit
algorithms, First Fit, Best Fit, Next Fit and Worst Fit can be found
in Knuth~\cite{Knuth97} and Wilson et al.~\cite{Wils95}.
These are closely related to problems of offline and online packing of
two-dimensional objects. One of the earliest considered packing variants is the problem of finding
a dense packing of a known set of squares for a rectangular container; see
Moser~\cite{m66}, Moon and Moser~\cite{mm67} and
Kleitman and Krieger~\cite{kk70}, as well as more recent work by
Novotn{\'y}~\cite{n95,n96} and Hougardy~\cite{h11}.
There is also a considerable number of other related work on offline packing squares, cubes, or hypercubes;
see~\cite{ck-soda04,js-ptas08,h09} for prominent examples.
The {\em online} version of square packing has been studied by
Januszewski and Lassak~\cite{jl97} and Han et al.~\cite{hiz08}, with more recent
progress due to Fekete and Hoffmann~\cite{fh-ossp-13,fh-ossp-17}.
A different kind of online square packing was considered by 
Fekete et al.~\cite{fks-osp-09,fks-ospg-14}. The container is an unbounded strip,
into which objects enter from above in a Tetris-like fashion; any new
object must come to rest on a previously placed object, and the 
path to its final destination must be collision-free. 

There are various ways to generalize the online packing of squares; see Epstein and van Stee~\cite{es-soda04,es05,es07} for online bin packing variants
in two and higher dimensions. In this context, also see parts of Zhang et al.~\cite{zcchtt10}.
A natural generalization of online packing of squares is online packing of rectangles,
which have also received a serious amount of attention. Most notably, online strip packing
has been considered; for prominent examples, see Azar and Epstein~\cite{ae-strip97}, who employ 
shelf packing, and Epstein and van Stee~\cite{es-soda04}.
Offline packing of rectangles into a unit square or rectangle has also been considered
in different variants; for examples, see \cite{fgjs05}, as well as \cite{jz-profit07}.
Particularly interesting for methods for online packing into a single container may be the work by Bansal et
al.~\cite{bcj-struct-09}, who show that for any complicated packing of rectangular items into a rectangular container,
there is a simpler packing with almost the same value of items. For another variant of online allocation, see~\cite{frs-csdaosa-14},
which extends previous work on optimal shapes for allocation~\cite{bbd-wosc-04}.

From within the FPGA community, there is a huge amount of related work
dealing with problems related to relocation.
Becker et al.~\cite{blc-erpbr-07} present a method for
enhancing the relocability of partial 
bitstreams for FPGA runtime configuration, with a special focus on
heterogeneities. They study the underlying prerequisites and
technical conditions for dynamic relocation. 
Gericota et al.~\cite{gericota05} present a relocation procedure for
Configurable Logic Blocks (CLBs) that is able to carry out online
rearrangements, defragmenting the available FPGA resources without
disturbing functions currently running. Another relevant approach was
given by Compton et al.~\cite{clckh-crdrt-02}, who present a new
reconfigurable architecture design extension based on the ideas of
relocation and defragmentation. 
Koch et al.~\cite{kabk-faepm-04} introduce efficient hardware extensions to
typical FPGA architectures in order to allow hardware task
preemption. 
These papers do not consider the algorithmic implications and how the
relocation capabilities can be exploited
to optimize module layout in a fast, practical fashion, which is what we consider in this paper. Koester et
al.~\cite{koester07} also address the problem of
defragmentation. Different defragmentation algorithms that minimize
different types of costs are analyzed. 

The general concept of defragmentation is well known, and has been applied to
many fields, e.g., it is typically employed for memory management. Our approach
is significantly different from defragmentation techniques which have been
conceived so far: these require a freeze of the system, followed by
a computation of the new layout and a complete reconfiguration of all modules
at once. Instead, we just copy one module 
at a time, and simply switch the execution to the new module as soon as the move is complete.
{\color{blue} This concept aims at providing a seamless, dynamic defragmentation
of the module layout, eventually resulting in much better
utilization of the available space for modules.} All this 
makes our work a two-dimensional extension of the one-dimensional approach 
described in \cite{fks-ddrd-12}.

\section{Preliminaries}

We are faced with an (online) sequence of configuration requests
that are to be carried out on a rectangular chip area. 
A request may consist of {\em deleting} an existing module, which 
simply means that the module may be terminated and
its occupied area can be released to free space.
On the other hand, 
a request may consist of {\em inserting} a new module,
requiring an axis-aligned, rectangular module
to be allocated to an unoccupied section of the chip;
if necessary, this may require rearranging the 
allocated modules in order to create free space of
the required dimensions, incurring some cost.

Previous work on reallocation problems of this type 
has focused on one-dimensional approaches. 
Using these in a two-dimensional setting does not result in 
satisfactory performance. 
The main contribution of our paper is to demonstrate a two-dimensional
approach that is able to achieve an efficiency that is provably within a constant factor of the optimum,
even in the worst case, which 
requires a variety of mathematical details. 
{\color{black} For better accessibility of the key ideas, 
our technical description in the rest of this Section~3,
as well as in Section~4 and Section~5 focuses on the case of quadratic modules
on a quadratic chip area. Section~6 addresses how to deal with general
rectangles.}


The rest of this section provides technical notation and descriptions.
A square is called \emph{aligned} if its edge length equals $2^{-r}$
for any $r \in \Nz$. It is called an $r$-square if its size is
$2^{-r}$ for a specific $r \in \Nz$. The \emph{volume} of an $r$-square $Q$ is $|Q|=4^{-r}$.
A {\em quadtree} is a rooted tree in which every node has either four
children or none. As a quadtree can be interpreted as the subdivision
of the unit square into nested $r$-squares, we can use quadtrees to
describe certain packings of aligned squares into the unit square. 

\begin{defi}
A \emph{(quadtree) configuration} $T$ assigns a set of axis-aligned squares to the
nodes of a quadtree. The nodes with a distance $j$ to the root of the
quadtree form \emph{layer} $j$. Nodes are also called \emph{pixels}
and pixels in layer $j$ are called \emph{$j$-pixels}. Thus,
$j$-squares can only be assigned to $j$-pixels. A
pixel $p$ \emph{contains} a square $s$ if $s$ is assigned to $p$ or
one of the children of $p$ contains $s$. A $j$-pixel that has
an assigned $j$-square is \emph{occupied}. 
For a pixel $p$ that is not occupied, with $P$ the unique path from $p$ to the root,
we call $p$ 
\begin{itemize}
\item \emph{blocked} if there is a $q \in P$ that is occupied,
\item \emph{free} if it is not blocked,
\item \emph{fractional} if it is free and contains a square,
\item \emph{empty} if it is free but not fractional,
\item \emph{maximally empty} if it is empty but its parent is not.
\end{itemize}

The \emph{height $h(T)$} of a configuration $T$ is defined as $0$ if the root of $T$ is empty. Otherwise, as the maximum $i+1$ such that $T$ contains an $i$-square.
\end{defi}

\begin{obs}\label{obs:disjoint}
Let $p \ne q$ be two maximally empty pixels and $P$ and $Q$ be the
paths from the root to $p$ and $q$, respectively. Then $p \notin Q$
and $q \notin P$.
\end{obs}
\begin{proof}
Without loss of generality, it is sufficient to show $p \notin Q$.
Assume $p \in Q$. Let $r \in Q$ be the parent of $q$. As $p$ is
maximally empty and $r$ is on the path from $p$ to $q$, $r$ must be
empty. However, that would imply that $q$ is not maximally empty, in
contradiction to the assumption.\qed
\end{proof}

	The \emph{(remaining) capacity $\mathrm{cap}(p)$} of a $j$-pixel $p$
is defined as $0$ if $p$ is occupied or blocked and as $4^{-j}$ if $p$ is empty. Otherwise, $\mathrm{cap}(p) := \sum_{p' \in C(p)} \mathrm{cap}(p')$, where $C(p)$ is the set of children of $p$. The \emph{(remaining)
capacity} of $T$, denoted $\mathrm{cap}(T)$, is the remaining
capacity of the root of $T$.

\begin{lem}\label{lem:fullcap}
Let $p_1, p_2, \ldots, p_k$ be all maximally empty pixels of a quadtree
configuration $T$. Then we have $\mathrm{cap}(T) = \sum_{i=1}^k \mathrm{cap}(p_i)$.
\end{lem}
\begin{proof}
The claim follows directly from the definition of the capacity, as the only
positive capacities considered for $\mathrm{cap}(T)$ are exactly those
of the maximally empty pixels. \qed
\end{proof}

\begin{figure}[tbh]
\centering
\includegraphics[width=0.45\textwidth]{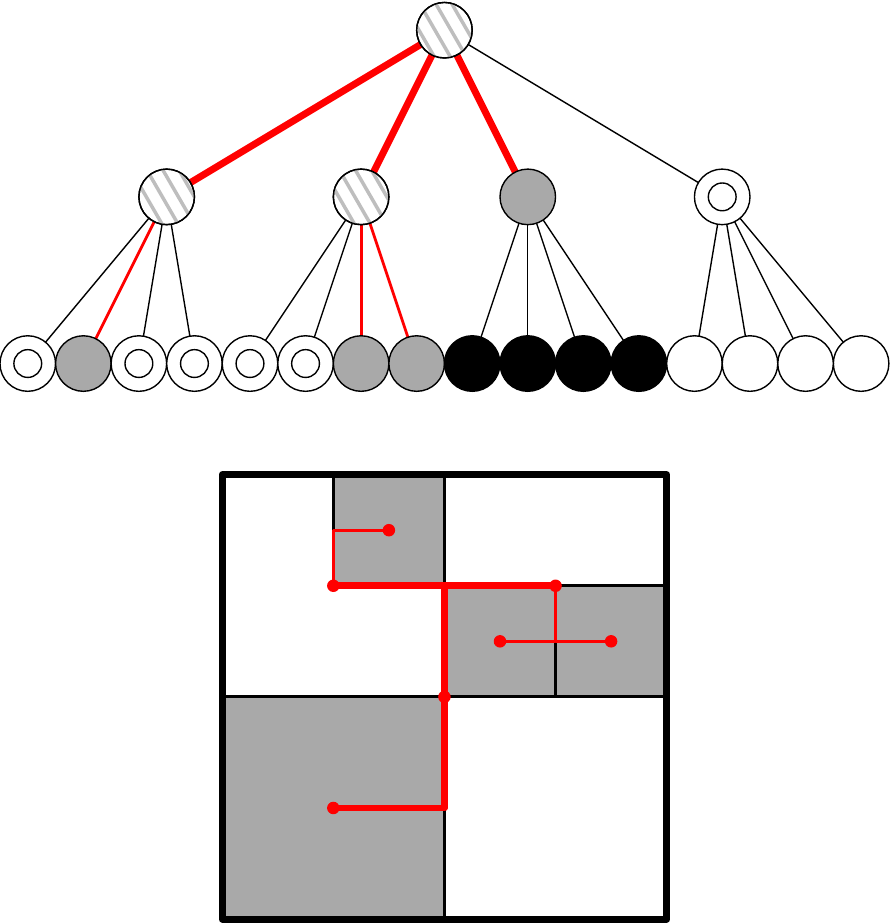}
\caption{A quadtree configuration {\color{black}(above)} and the corresponding dynamically generated quadtree layout {\color{black}(below)}.
Gray nodes are occupied, white ones with gray stripes
fractional, black ones blocked, and white nodes without stripes
empty. Maximally empty nodes have a circle inscribed. Red lines in the module layout
indicate the dynamically produced communication infrastructure, induced by the quadtree structure.}
\label{fig:config}
\end{figure}

See Figure~\ref{fig:config} for an example of a quadtree configuration
and the corresponding packing of aligned squares in the unit square.
	
	Quadtree configurations are transformed using \emph{moves}
(\emph{reallocations}). A $j$-square $s$ assigned to a $j$-pixel
$p$ can be \emph{moved} (\emph{reallocated}) to another $j$-pixel $q$
by creating a new assignment from $q$ to $s$ and deleting the old
assignment from $p$ to $s$. $q$ must have been empty for this to be
allowed.

We allow only one move at a time. For example, two squares cannot
change places unless there is a sufficiently large pixel to
temporarily store one of them. Furthermore, we do not put limitations
on how to transfer a square from one place to another, i.e., we can
always move a square even if there is no collision-free path between
the origin and the destination.

\begin{defi}
A fractional pixel is \emph{open} if at least one of its children is
(maximally) empty. A configuration is called \emph{compact} if there
is at most one open $j$-pixel for every $j \in \Nz$.
\end{defi}

In (one-dimensional) storage allocation and scheduling, there are
techniques that avoid reallocations by requiring more space than the
sum of the sizes of the allocated
pieces. See Bender et al.~\cite{bender_reallocation_2014} for an
example. From there we adopt the term \emph{underallocation}. In particular, given two squares $s_1$ and $s_2$, $s_2$ is an $x$-underallocated copy
of $s_1$, if $|s_2| = x \cdot |s_1|$ for $x > 1$.

\begin{defi}
A \emph{request} has one of the forms \textsc{Insert($x$)} or
\textsc{Delete($x$)}, where $x$ is a unique identifier for a
square. Let $v \in [0, 1]$ be the volume of the square $x$. The
\emph{volume} of a request $\sigma$ is defined as
\[
\mathrm{vol}(\sigma) = \left\{ \begin{array}{ccl}
 v & \text{if} & r=\textsc{Insert($x$)},\\
-v & \text{if} & r=\textsc{Delete($x$)}.
\end{array} \right.
\]
\end{defi}

\begin{defi}
A sequence of requests $\sigma_1, \sigma_2, \ldots, \sigma_k$
is \emph{valid} if $\sum_{i=1}^j \mathrm{vol}(\sigma_i) \le 1$ holds
for every $j=1,2,\ldots,k$. It is called \emph{aligned}, if
$|\mathrm{vol}(\sigma_j)| = 4^{-\ell_j}, \ell_j \in \Nz,$
where $|.|$ denotes the absolute value, holds for every
$j=1,2,\ldots,k$, i.e., if only aligned squares are packed.
\end{defi}

Our goal is to minimize the costs of reallocations. Costs can be
measured in different ways, for example in the number of moves or the
reallocated volume.

\begin{defi}
Assume we fulfill a request $\sigma$ and as a consequence reallocate
a set of squares $\{s_1, s_2, \ldots, s_k\}$. The \emph{movement
cost} of $\sigma$ is defined as $c_{\mathrm{move}}(\sigma) = k$,
the \emph{total volume cost} of $\sigma$ is defined as
$c_{\mathrm{total}}(\sigma) = \sum_{i=1}^k |s_i|$,
and the \emph{(relative) volume cost} of $\sigma$ is defined as
$c_{\mathrm{vol}}(\sigma) = \frac{c_{\mathrm{total}}(\sigma)}{|\mathrm{vol}(\sigma)|}$.
\end{defi}

\section{Inserting into a Given Configuration}

In this section we examine the problem of rearranging a given
configuration in such a way that the insertion of a new square is
possible. {\color{black}
Before we present our results in mathematical detail, including all
necessary proofs, we give a short overview of the individual
propositions and their significance: We first examine properties of
quadtree configurations culminating in Theorem~\ref{thm:qtmoves}, which
establishes that any configuration with sufficient capacity allows the
insertion of a square. Creating the required contiguous space for the insertion
comes at a cost due to required reallocations. This cost is analysed
in detail in Subsection~\ref{sec:costs}. There, we present matching
upper and lower bounds on the reallocation cost for our three cost
functions -- total volume cost (Theorems~\ref{thm:volumebound} and
\ref{thm:volumexample}), (relative) volume cost
(Corollary~\ref{cor:relvoltight}), and movement cost
(Theorems~\ref{thm:movbound} and \ref{thm:movexample}).
}

\subsection{Coping with Fragmented Allocations}\label{sec:delete}

Our strategy follows one general idea: larger empty pixels can be
built from smaller ones; e.g., four empty $i$-pixels can
be combined into one empty $(i-1)$-pixel. This can be iterated
to build an empty pixel of suitable volume.

\begin{lem}\label{lem:order}
Let $p_1, p_2, \ldots, p_k$ be a sequence of empty pixels sorted by
volume in descending order. Then
$\sum_{i=1}^k \mathrm{cap}(p_i) \ge 4^{-\ell} > \sum_{i=1}^{k-1} \mathrm{cap}(p_i)$
implies the following properties:
\begin{equation}\label{eq:four_p1}
k < 4 \Leftrightarrow k = 1
\end{equation}
\begin{equation}\label{eq:exact}
k \ge 4 \Rightarrow \sum_{i=1}^k \mathrm{cap}(p_i) = 4^{-\ell}
\end{equation}
\begin{equation}\label{eq:four_p2}
k \ge 4 \Rightarrow \mathrm{cap}(p_k) = \mathrm{cap}(p_{k-1}) =
\mathrm{cap}(p_{k-2}) = \mathrm{cap}(p_{k-3})
\end{equation}
\end{lem}

\begin{proof}
For $k \ge 2$, $p_1$ must be a pixel of smaller capacity than an
$\ell$-pixel, because otherwise we would not need $p_2$ for the sum to
be greater than $4^{-\ell}$ -- in contradiction to the
assumption. Thus, we need to add up smaller capacities to at least
$4^{-\ell}$. As we need at least four $(\ell+1)$-pixels for that,
statement~\eqref{eq:four_p1} holds.

In the following we assume $k \ge 4$. Let $x=\sum_{i=1}^{k-1}
\mathrm{cap}(p_i)$. We know from the assumption that $x$ is strictly
less than $4^{-\ell}$, but $x+\mathrm{cap}(p_k)$ is at least
$4^{-\ell}$. Consider the base-4 (quaternary) representation
of $x/4^{-\ell}$: $x_4=(x/4^{-\ell})_4$. It has a zero before the
decimal point and a sequence of base-4 digits after. Let $n$ be the
rightmost non-zero digit of $x_4$. As the sequence is sorted in
descending order and the capacities are all negative powers of four,
adding the capacity of $p_k$ can only increase $n$, or a digit right
of $n$, by one. Since all digits right of $n$ are zero, increasing one
of them by one does not increase $x$ to at least
$4^{-\ell}$. Therefore, it must increase $n$. But if increasing $n$ by
one means increasing $x$ to at least $4^{-\ell}$, then every digit of
$x_4$ after the decimal point and up to $n$ must have been
three. Consequently, increasing $n$ by one leads not only to
$x + \mathrm{cap}(p_k) \ge 4^{-\ell}$ but also to
$x + \mathrm{cap}(p_k)=4^{-\ell}$, which is statement~\eqref{eq:exact}.

Furthermore, as $n$ must have been three and the sequence is sorted,
the previous three capacities added must have each increased $n$ by
exactly one as well. This proves statement~\eqref{eq:four_p2}. \qed
\end{proof}

\begin{figure}[t!hp]
\centering
\includegraphics[width=0.25\textwidth]{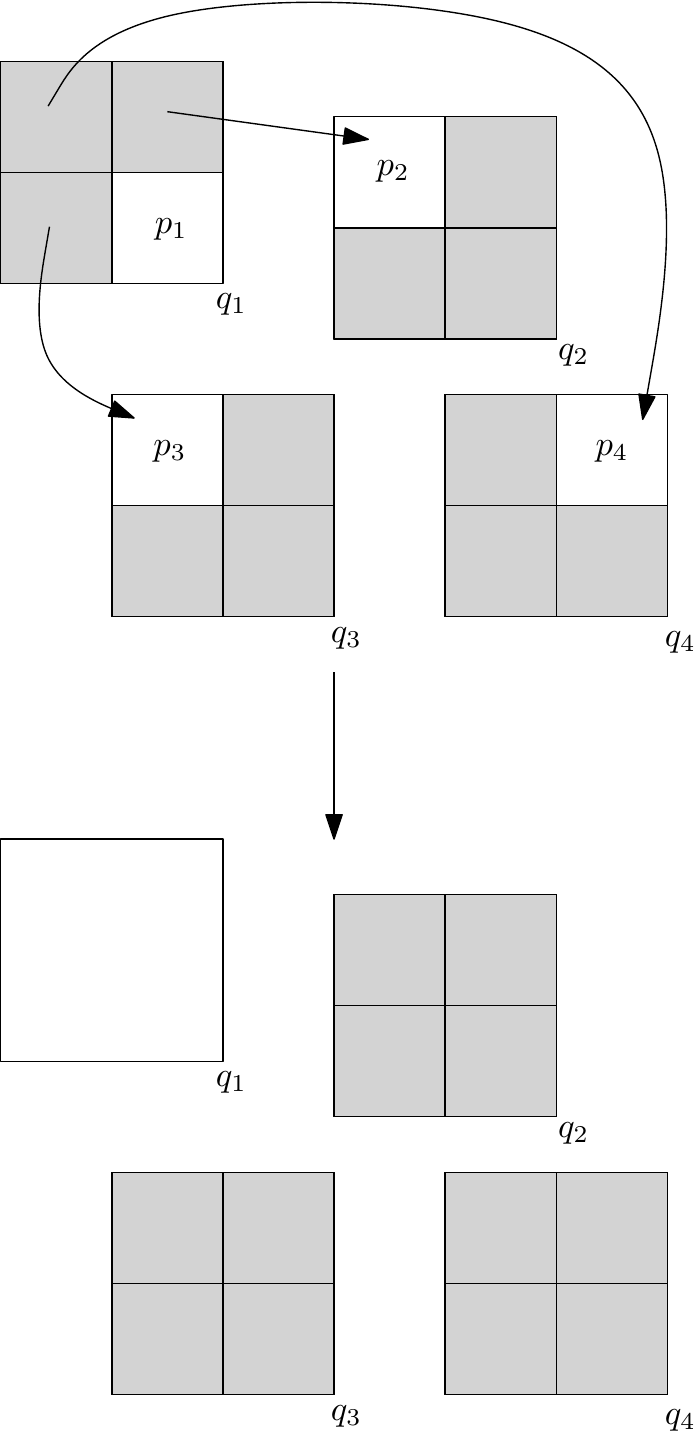}
\caption{Illustration to Lemma~\ref{lem:four}.}
\label{fig:four_pixels}
\end{figure}

\begin{lem}\label{lem:four}
Given a quadtree configuration $T$ with four maximally empty
$j$-pixels. Then $T$ can be transformed (using a sequence
of moves) into a configuration $T^*$ with one more maximally
empty $(j-1)$-pixel and four fewer maximally empty $j$-pixels than $T$
while retaining all its maximally empty $i$-pixels for $i < j-1$.
\end{lem}

\begin{proof}
Let $p_1, p_2, p_3$ and $p_4$ be four maximally empty $j$-pixels and
$q_1, q_2, q_3$ and $q_4$ be the parents of $p_1, p_2, p_3$ and $p_4$,
respectively. Then $q_i$ has at most three children that are not
empty. Now, we can move the at most three non-empty subtrees from one
of the $q_i$ to the others, $i=1,2,3,4$. Without loss of generality,
we choose $q_1$. Let $a, b$ and $c$ be the children of $q_1$ that are
not $p_1$. We move $a$ to $p_2$, $b$ to $p_3$ and $c$ to $p_4$. See
Figure~\ref{fig:four_pixels} for an illustration. Thus, we get a new
configuration $T^*$ with the empty $(j-1)$-pixel $q_1$ and occupied or
fractional pixels {\color{black}$q_2$, $q_3$, $q_4$}. Note that $p_1$ is still empty,
but no longer maximally empty, because its parent $q_1$ is now empty.
The construction does not affect any other maximally empty pixels. \qed
\end{proof}

\begin{thm}\label{thm:qtmoves}
Given a quadtree configuration $T$ with a remaining capacity of at least
$4^{-j}$, you can transform $T$ into a quadtree configuration $T^*$
with an empty $j$-pixel using a sequence of moves.
\end{thm}

\begin{proof}
Let $S=p_1, p_2, \ldots, p_n$ be the sequence containing all maximally
empty pixels of $T$ sorted by capacity in descending order. If the
capacity of $p_1$ is at least $4^{-j}$, then there already is an empty
$j$-pixel in $T$ and we can simply set $T^* = T$.

Assume $\mathrm{cap}(p_1) < 4^{-j}$. In this case we inductively build an empty
$j$-pixel. Let $S'=p_1, p_2, \ldots, p_k$ be the shortest prefix of
$S$ satisfying $\sum_{i=1}^{k} \mathrm{cap}(p_i) \ge 4^{-j}$.
Such a prefix has to exist because of {\color{black}Lemma~\ref{lem:fullcap}}.
Note that due to Observation~\ref{obs:disjoint} no pixel
$p_i$ is contained in another pixel $p_j$, $i, j \in
\{1,2,\ldots,k\}$, $i \ne j$.
Lemma~\ref{lem:order} tells us $k \ge 4$ and the
last four pixels in $S'$, $p_{k-3}, p_{k-2}, p_{k-1}$ and
$p_k$, are from the same layer, say layer $\ell$. Thus, we can apply
Lemma~\ref{lem:four} to $p_{k-3}, p_{k-2}, p_{k-1}, p_k$ to get a new
maximally empty $(\ell - 1)$-pixel $q$. We remove $p_{k-3}, p_{k-2},
p_{k-1}, p_k$ from $S'$ and insert $q$ into $S'$ according to its
capacity. 
The length of the resulting sequence $S''$ is three less than
the length of $S'$. This does not change the sum of the capacities, since
an empty $(\ell-1)$-pixel has the same capacity as four empty
$\ell$-pixels. That is, $\sum_{p \in S'} \mathrm{cap}(p) = \sum_{p \in
 S''} \mathrm{cap}(p)$ holds.

We can repeat these steps until $k < 4$ holds. Then
Lemma~\ref{lem:order} implies that $k=1$, i.e., the sequence contains
only one pixel $p_1$, and because $\mathrm{cap}(p_1)=4^{-j}$, $p_1$ is
an empty $j$-pixel. \qed
\end{proof}

\subsection{Reallocation Cost}\label{sec:costs}

Reallocation cost is made non-trivial by \emph{cascading moves}:
Reallocated squares may cause further reallocations, when there is no
empty pixel of the required size available.

\begin{obs}\label{obs:largebad}
In the worst case, reallocating an $\ell$-square is not cheaper than
reallocating four $(\ell+1)$-squares -- using any of the three defined
cost types.
\end{obs}

\begin{proof}
It is straightforward to see this for volume costs, total or relative:
Wherever you can move one $\ell$-square you can also move four
$(\ell+1)$-squares without causing more cascading moves.

For movement costs a single move of an $\ell$-square is less than four
moves of $(\ell+1)$-squares, but it can cause cascading moves of three
$(\ell+1)$-squares plus the cascading moves caused by the reallocation
of an $(\ell+1)$-square and, therefore, does not cause lower costs in
total. \qed
\end{proof}

\begin{thm}\label{thm:volumebound}
The maximum total volume cost caused by the insertion of an
$i$-square $Q$, $i \in \Nz$, 
into a quadtree configuration $T$ with
$\mathrm{cap}(T) \ge 4^{-i}$ is bounded by
\[
c_\mathrm{total,max} \le \frac{3}{4} \cdot
4^{-i} \cdot \mathrm{min} \{(s-i), i\} \in O(|Q| \cdot h(T))
\]
when the smallest previously inserted square is an $s$-square.
\end{thm}

\begin{proof}
For $s \le i$ there has to be an empty $i$-square in $T$, as
$\mathrm{cap}(T) \ge 4^{-i}$, and we can insert $Q$ without any moves. In
the following, we assume $s > i$.
Let $Q$ be the $i$-square to be inserted. We can assume that we do not
choose an $i$-pixel with a remaining capacity of zero to pack $Q$ -- if
there were no other pixels, $\mathrm{cap}(T)$ would be zero as
well. Therefore, the chosen pixel, say $p$, must have a remaining
capacity of at least $4^{-s}$. From Observation~\ref{obs:largebad}
follows that the worst case for $p$ would be to be filled with 3
$k$-squares, for every $i < k \le s$. Let $v_i$ be the worst-case
volume of a reallocated $i$-pixel. We get $v_i \le \sum_{j=i+1}^s \frac{3}{4^j} = 4^{-i} - 4^{-s}$.

Now we have to consider cascading moves. Whenever we move an
$\ell$-square, $\ell > i$, to make room for $Q$, we might
have to reallocate a volume of $v_{\ell}$ to make room for the
$\ell$-square. Let $x_i$ be the total volume that is at most
reallocated when inserting an $i$-square. 
Then we get the recurrence $x_i = v_i + \sum_{j=i+1}^s 3 \cdot x_j$
with $x_s=v_s=0$. This resolves to $x_i=3/4 \cdot 4^{-i} \cdot (s-i)$.

$v_i$ cannot get arbitrarily large, as the
remaining capacity must suffice to insert an $i$-square. Therefore,
if all the possible $i$-pixels contain a volume of $4^{-s}$ (if some
contained more, we would choose those and avoid the worst case), we
can bound $s$ by $4^i \cdot 4^{-s} \ge 4^{-i} \Leftrightarrow s \le 2i$,
which leads to $c_\mathrm{total,max} \le \frac{3}{4} \cdot 4^{-i} \cdot i$.

With $|Q|=4^{-i}$ and $i < s < h(T)$ we get $c_\mathrm{total,max} \in
O(|Q| \cdot h(T))$. \qed
\end{proof}

\begin{cor}\label{cor:maxtotalvol}
Inserting a square into a quadtree configuration has a total volume
cost of no more than $3/16=0.1875$.
\end{cor}

\begin{proof}
Looking at Theorem~\ref{thm:volumebound} it is easy to see that the
worst case is attained for $i=1$: $c_\mathrm{total}
= 3/4 \cdot 4^{-1} \cdot 1 = 3/16=0.1875$. \qed
\end{proof}

\begin{figure}[h!]
\centering
\includegraphics[width=0.40\textwidth]{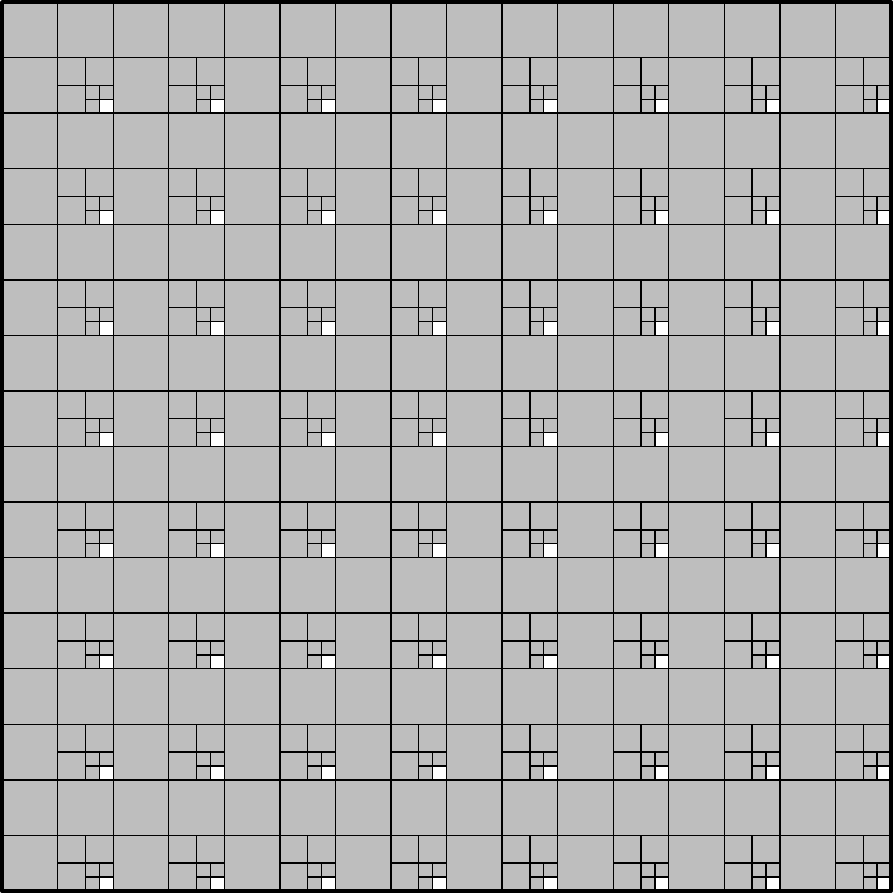}
\caption{The worst-case construction for volume cost for $s=6$ and
  $i=3$. Every 3-pixel contains three 4-, 5-, and 6-squares with only
  one remaining empty 6-pixel.}
\label{fig:wc_volume}
\end{figure}

\begin{thm}\label{thm:volumexample}
For every $i \in \Nz$ there are quadtree configurations $T$ for which the
insertion of an $i$-square $Q$ causes a total volume
cost of
\[
c_\mathrm{total,max} \ge \frac{3}{4} \cdot
4^{-i} \cdot \mathrm{min} \{(s-i), i\} \in \Omega(|Q| \cdot h(T))
\]
when the smallest previously inserted square is an $s$-square.
\end{thm}

\begin{proof}
We build a quadtree configuration to match the upper bound of
Theorem~\ref{thm:volumebound}. Let $s=2i$ and consider a subtree
rooted at an $i$-pixel
that contains three $k$-pixels for every $i < k \le s$. They do not have
to be arranged in such a way that the single free $s$-pixel is in the
lower right corner, but the nesting structure is important. Assume all
$4^i$ $i$-pixels of $T$ are constructed in such a way. Then you have
to reallocate three $k$-squares for every $i < k \le s$. However, every
fractional $k$-pixel in the configuration in turn contains three
$k'$-pixel for every $k < k' < s$, i.e., moving every $k$-square
causes cascading moves. See Figure~\ref{fig:wc_volume} for the whole
construction for $s=6$ and $i=3$. The reallocated volume without
cascading moves adds up to $v_i = \sum_{k=i+1}^s 3 \cdot 4^{-k}$.

Including cascading moves we get $x_i = v_i + \sum_{k=i+1}^s 3 \cdot x_k$,
which resolves to $x_i=3/4 \cdot 4^{-i} \cdot (s-i)$.

With
$s=h(T)-1$, $i=s/2$ and $|Q|=4^{-i}$ we get
$c_\mathrm{total,max} \in \Omega(|Q| \cdot h(T))$. \qed
\end{proof}

As a corollary we get an upper bound for the (relative) volume cost
and a construction matching the bound.

\begin{cor}\label{cor:relvoltight}
Inserting an $i$-square into a quadtree configuration $T$ with sufficient capacity
$\mathrm{cap}(T) \ge 4^{-i}$ causes a (relative) volume cost of at
most
\[
c_\mathrm{vol,max} \le \frac{3}{4} \cdot \mathrm{min} \{(s-i), i\} \in
\Theta(h(T)),
\]
when the smallest previously inserted square is an $s$-square,
and this bound is tight, i.e., there are configurations for which the
bound is matched.
\end{cor}

It is important to note that
relative volume cost can be arbitrarily bad by increasing the height
of the configuration, as opposed to total volume cost with the upper
bound derived in
Corollary~\ref{cor:maxtotalvol}. What is more, large total volume
cost is achieved by inserting $i$-squares for small $i$, whereas
large relative volume cost is only possible for large $i$ (and large
$s-i$). This has an interesting interpretation with regard to the structure of
the quadtree: Large total volume cost can happen when you assign a
square to a node close to the root. To get large relative volume cost
you need a high quadtree and assign a square to a node roughly in the
middle (with respect to height).

The same methods we used to derive worst case bounds for volume cost
can also be used to establish bounds for movement cost, which results
in $c_\mathrm{move,max} \le 4^{\mathrm{min}\{s-i, i\}}-1 \in
O(2^{h(T)})$. A matching construction is the same as the one in the
proof of Theorem~\ref{thm:volumexample}.

\begin{thm}\label{thm:movbound}
The maximum movement cost caused by the insertion of an
$i$-square $Q$, $i \in \Nz$, into a quadtree configuration $T$ with
$\mathrm{cap}(T) \ge 4^{-i}$ is bounded by
\[
c_\mathrm{move,max} \le 4^{\mathrm{min}\{s-i, i\}}-1 \in O(2^{h(T)})
\]
when the smallest previously inserted square is an $s$-square.
\end{thm}
\begin{proof}
The proof is analogous to the proof of
Theorem~\ref{thm:volumebound}. We can use
Observation~\ref{obs:largebad} and formulate a new recurrence. The
number of reallocations without cascading moves caused by the
insertion of $Q$ can be bounded by $v_i \le 3(s-i)$ and including
cascading moves we get $x_i = v_i + \sum_{j=i+1}^s 3 x_i$, which
resolves to $x_i = 4^{s-i} - 1$.

As we need at least $4^{-i}$ remaining capacity to insert $Q$ we can
again deduce $s \le 2i$. With $s=h(T)-1$ we get $\mathrm{min}\{s-i,
i\} \le h(T)/2$, which results in the claimed bound. \qed
\end{proof}

\begin{thm}\label{thm:movexample}
For every $i \in \Nz$ there are quadtree configurations $T$ for which the
insertion of an $i$-square $Q$ causes a movement
cost of
\[
c_\mathrm{move,max} \ge 4^{\mathrm{min}\{s-i, i\}} \in \Omega(2^{h(T)})
\]
when the smallest previously inserted square is an $s$-square.
\end{thm}
\begin{proof}
The example from Theorem~\ref{thm:volumexample} works here as well. As
every fractional $j$-pixel, $j < s$, contains three $(j+1)$-pixels,
you have to move three squares for every $j=i,\ldots,s-1$ and account
for cascading moves. This results in a number of moves $c_{move,max}
\ge x_i = 3(s-i) + \sum_{j=i+1}^s x_j = 4^{s-i} - 1$, where
$s=2i=h(T)-1$. \qed
\end{proof}

\section{Online Packing and Reallocation}

Applying Theorem~\ref{thm:qtmoves} repeatedly to successive configurations
yields a strategy for the dynamic allocation of aligned squares.

\begin{cor}\label{cor:strategy}
Starting with an empty square and given a valid, aligned sequence of
requests, there is a strategy that fulfills every request in the
sequence.
\end{cor}
\begin{proof}
We only have to deal with aligned squares and can use quadtree
configurations to pack the squares, since the sequence of requests
$\sigma_1, \sigma_2, \ldots, \sigma_k$ is aligned. We start with the
empty configuration that contains only one empty $0$-pixel. Thus, we
have a configuration with capacity $1$. We only have to consider
insertions, because deletions can always be fulfilled by definition.

As the sequence of requests is valid, whenever a request $\sigma_\ell$
demands to insert a $j$-square $s$, the remaining capacity of the
current quadtree configuration $T$ is at least
$1 - \sum_{i=1}^{\ell-1} \mathrm{vol}(\sigma_i) + 4^{-j} \ge 4^{-j}$.

Therefore, we can use Theorem~\ref{thm:qtmoves} to transform $T$ into
a configuration $T^*$ with an empty $j$-pixel $p$. We assign $s$
to $p$. \qed
\end{proof}

This strategy may incur the heavy insertion cost
  derived in the previous section. However, when we do not have to
  work with a given configuration and have the freedom to handle all
  requests starting from the empty unit square, we can use the added
  flexibility to derive a more sophisticated strategy. In particular,
  we can use reallocations to clean up a configuration when squares
  are deleted. This can make deletions costly operations, but allows us
  to eliminate insertion cost entirely.

\subsection{First-Fit Packing}

We present an algorithm that fulfills any valid, aligned sequence of
requests and does not cause any reallocations on insertions. We call
it \emph{First Fit} in imitation of the well-known technique employed
in one-dimensional allocation problems.

Given a one-dimensional
packing and a request to allocate space for an additional item,
First-fit chooses the first suitable location. In one dimension it is
trivial to define an order in which to check possible locations. For
example, assume your resources are arranged horizontally and proceed
from left to right.


{\color{black}
Finding an order in two or more dimensions
is less straightforward than in 1D. We use space-filling curves to overcome this
impediment. Space-filling curves are of theoretical interest, because
they fill the entire unit square
(i.e., their Hausdorff dimension is $2$). 
More useful for us are the schemes used to create a
space-filling curve, which employ a recursive construction on the nodes
of a quadtree and become space-filling as the height of the tree
approaches infinity. In particular, they provide an order for the
nodes of a quadtree. In the following, we make use of the z-order
curve~\cite{morton_1966}.
}

First Fit assigns items to be packed to the next available position in
z-order. We denote the position
of a pixel $p$ in z-order by $z(p)$, i.e.,
$z(p) < z(q)$ if and only if $p$ comes before $q$ in z-order.

In
general, the z-order is only a partial order, as it does not make
sense to compare nodes with their parents or children. However, there
are three important occasions for which the z-order is a total order: If
you only consider pixels in one layer, if you only consider
occupied pixels, and if you only consider maximally empty pixels. In
all three cases pixels are pairwise disjoint, which
leads to a total order.

\begin{figure}[h!]
\centering
\includegraphics[width=0.45\textwidth]{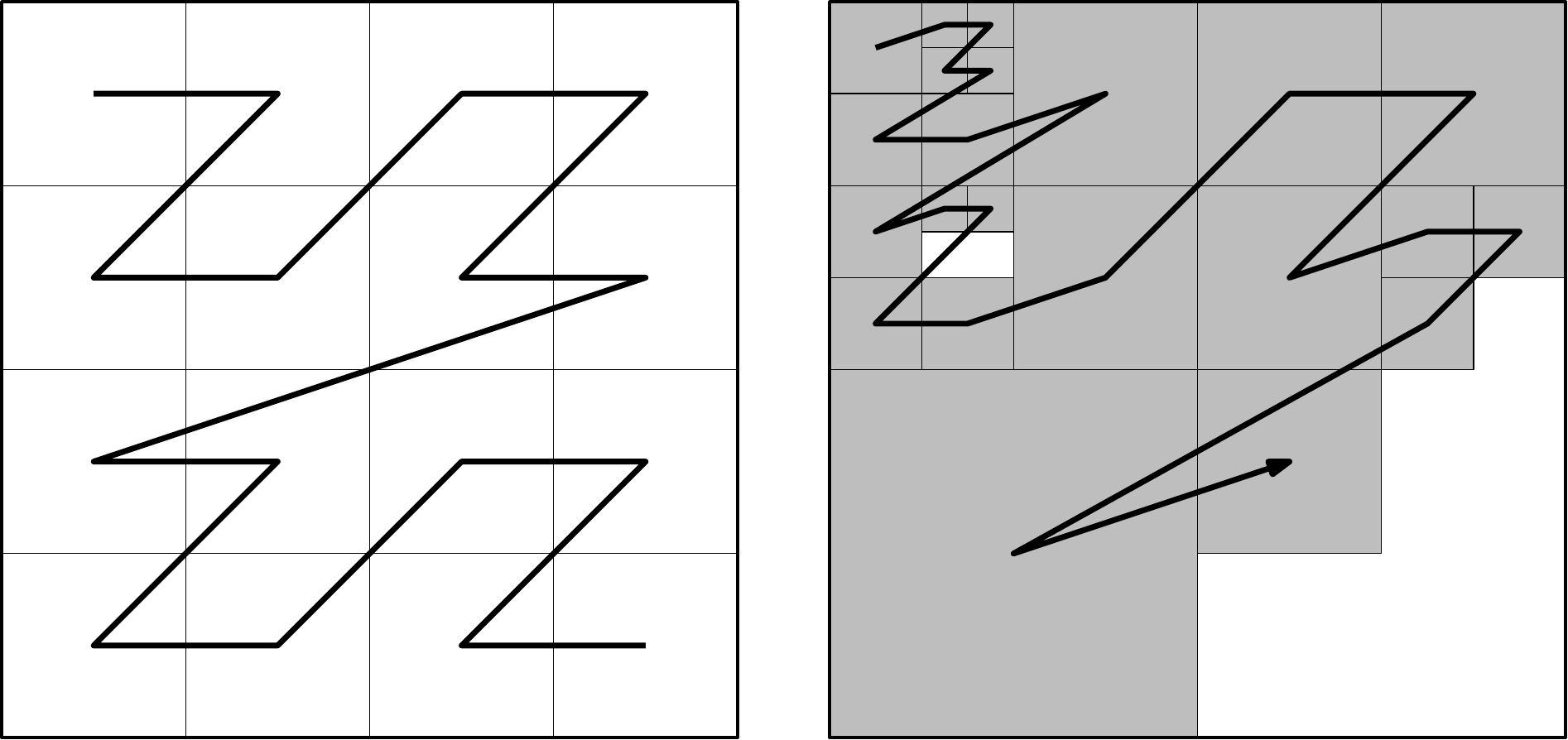}
\caption{The z-order for layer 2 pixels (left); a First Fit allocation and the z-order of the occupied pixels
-- which is not necessarily the insertion order (right).}
\label{fig:firstfit}
\end{figure}

First Fit proceeds as follows: A request to insert an $i$-square
$Q$ is handled by assigning $Q$ to the first empty $i$-pixel in
z-order; see Figure~\ref{fig:firstfit}. Deletions are more
complicated. After unassigning a
deleted square $Q$ from a pixel $p$ the following procedure handles
reallocations (an example deletion can be seen in
Figure~\ref{fig:invstrategy}):

\begin{algorithmic}[1]
\State $S \gets \{p'\}$, where $p'$ is the maximally empty pixel
       containing $p$
\While{$S \ne \varnothing$}
  \State Let $a$ be the element of $S$ that is first in z-order.
  \State $S \gets S \setminus \{a\}$
  \State Let $b$ be the last occupied pixel in z-order.
  \While{$z(b) > z(a)$}
    \If{the square assigned to $b$, $B$, can be packed into $a$}
      \State Assign $B$ to the first suitable descendant of $a$ in
             z-order.
      \State Unassign $B$ from $b$.
      \State Let $b'$ be the maximally empty pixel containing $b$.
      \State $S \gets S \cup \{b'\}$
      \State $S \gets S \setminus \{b'': b''\text{ is child of }b'\}$
    \EndIf
    \State Move the pointer $z$ back in z-order to the next occupied
           pixel.
  \EndWhile
\EndWhile
\end{algorithmic}

The general idea is to reallocate squares from the current end of the
z-order to empty spots. As reallocating creates new empty squares, we
need to apply the method repeatedly in what can be considered an
inverse case of cascading moves. We ensure termination by always
moving the currently {\color{black}considered} empty pixel in positive z-order and
reallocating squares in negative z-order. We analyze the strategy in
more detail now.

\begin{inv}\label{inv:inv}
For every empty $i$-pixel $p$ in a quadtree configuration $T$ there is
no occupied $i$-pixel $q$ with $z(q) > z(p)$.
\end{inv}

\begin{lem}\label{lem:invcompact}
Every quadtree configuration $T$ satisfying Invariant~\ref{inv:inv}
is compact.
\end{lem}

\begin{proof}
Assume a quadtree configuration $T$ is not compact. Then it contains
two fractional $i$-pixels, $i \in \N$, $p$ and $q$ with maximally
empty children
$p'$ and $q'$, respectively. Without loss of generality, assume $z(p)
< z(q)$. As $q$ is fractional, there is a $j$-square, $j > i$, assigned
to some descendant of $q$, say $q''$. However, $p'$ is an empty
$(i+1)$-pixel and therefore contains an empty $j$-pixel, $p''$. As
$z(p) < z(q)$, we also have $z(p'') < z(q'')$ and
Invariant~\ref{inv:inv} does not hold. \qed
\end{proof}

\begin{lem}\label{lem:3max}
In a compact quadtree configuration $T$ there are at most three
maximally empty $j$-pixels for every $j \in \Nz$.
\begin{proof}
The statement holds for $j=0$, since there is only one $0$-pixel. For
$j>0$ there is at most one open $(j-1)$-pixel $p$ in $T$, because $T$
is compact. Therefore, all other $(j-1)$-pixels except for $p$ either
do not have an empty child or are maximally empty themselves. Thus,
all maximally empty $j$-pixels have to be children of $p$. Since $p$
is not empty, there can be at most three. \qed
\end{proof}
\end{lem}

\begin{lem}\label{lem:compactspace}
Given an $\ell$-square $s$ and a compact quadtree configuration $T$,
then $s$ can be assigned to an empty $\ell$-pixel in $T$, if and only
if $\mathrm{cap}(T) \ge 4^{-l}$.
\end{lem}

\begin{proof}
The direction from left to right is obvious, as there can be no empty
$\ell$-pixel if the capacity is less than $4^{-l}$. For the other
direction assume there is no empty $\ell$-pixel in $T$. Since
there is no empty $\ell$-pixel, there is also no empty $j$-pixel for
any $j < \ell$. Let the smallest square assigned to a node be an
$s$-square. As $T$ is compact, we can use Lemma~\ref{lem:3max}
and {\color{black}Lemma~\ref{lem:fullcap}} to bound the remaining capacity of $T$ from
above: $\mathrm{cap}(T) \le \sum_{k=l+1}^{s} 3 \cdot 4^{-k} = 4^{-\ell} -
4^{-s} < 4^{-\ell}$. \qed
\end{proof}

In other words, packing an $\ell$-square in a compact configuration
requires no reallocations.

\begin{thm}\label{thm:ff}
The strategy presented above is correct. In particular,
\begin{enumerate}
\item every valid insertion request is fulfilled at zero cost,
\item every deletion request is fulfilled,
\item after every request Invariant~\ref{inv:inv} holds.
\end{enumerate}
\end{thm}

\begin{proof}
The first part follows from Lemmas~\ref{lem:compactspace}
and \ref{lem:invcompact} and point 3. Insertions maintain the invariant,
because we assign it to the first suitable empty
pixel in z-order. Deletions can obviously always be fulfilled. We
still need to prove the important part, which is that the invariant
holds after a deletion.

We show this by proving that whenever the procedure reaches line 3 and
sets $a$, the invariant holds for all squares in
z-order up to $a$. As we only move squares in negative z-order, the
sequence of pixels $a$ refers to is increasing in z-order. Since we
have a finite number of squares, the procedure terminates after a
finite number of steps when no suitable $a$ is left. At that point the
invariant holds throughout the configuration.

Assume we are at step 3 of the procedure and the invariant holds for
all squares up to $a$. None of the squares considered to be moved to
$a$ fit anywhere before $a$ in z-order -- otherwise the invariant
would not hold for pixels before $a$. Afterwards, no square that has
not been moved to $a$ fits into $a$, because it would have been moved
there otherwise. Once we reach line 3 again, and set the new $a$, say
$a'$, consider the pixels between $a$ and $a'$ in z-order. If any
square after $a'$ would fit somewhere into a pixel between $a$ and
$a'$, then the invariant would not have held before the
deletion. Therefore, the invariant holds up to $a'$. \qed
\end{proof}

{\color{black}
Comparing our results in Section~4 to those in this section, a major
advantage of an empty initial configuration becomes apparent. For all
examined cost functions there are configurations into which no square
can be inserted at zero cost (cf. Theorem~\ref{thm:volumexample},
Corollary~\ref{cor:relvoltight}, Theorem~\ref{thm:movexample}). This
is in contrast to First Fit, which achieves insertion at zero
cost (Theorem~\ref{thm:ff}). The downside is the potentially large cost
of deletions. The thorough analysis of a strategy with provably low
cost for both insertions and deletions is the subject of future work.
}

\begin{figure}[h!]
\centering
\includegraphics[width=0.25\textwidth]{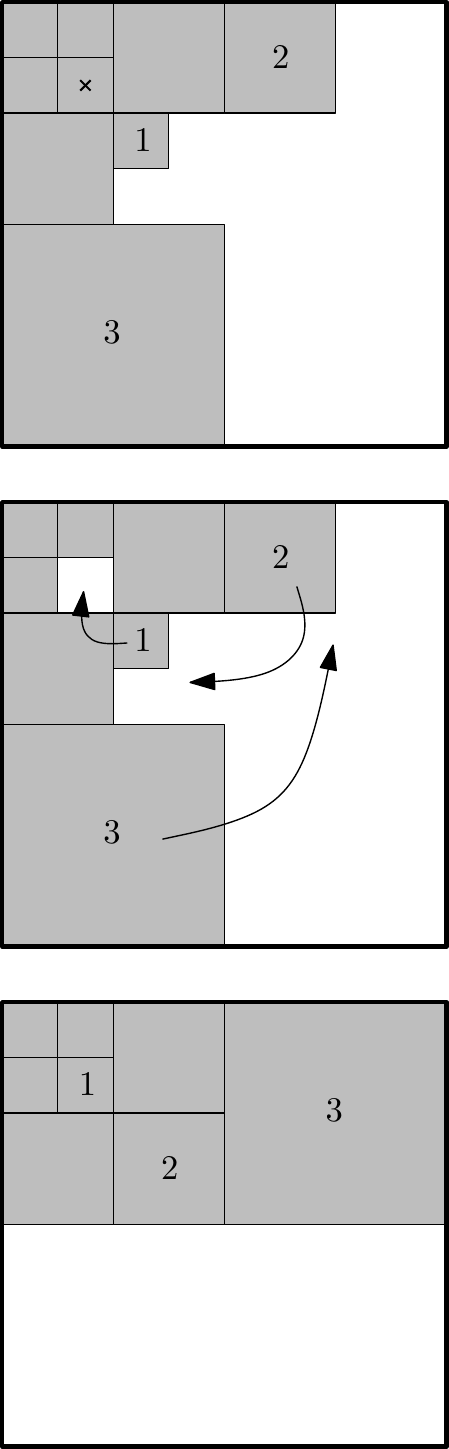}
\caption{Deleting a square causes several moves. The
deleted square is marked with a cross. Once it is unassigned, the
squares are checked in reverse z-order until square 1, which
fits. Afterwards, there is a now maximally empty pixel into which
square 2 can be moved. Finally, the same happens for square 3.}
\label{fig:invstrategy}
\end{figure}

\section{General Squares and Rectangles}\label{sec:generals}

Due to limited space and for clearer exposition,
the description in the previous three sections considered aligned squares. 
We can adapt the technique to general squares
and even rectangles at the expense of a constant factor.


{\color{black}
To accommodate a non-aligned square, we pack it like an
aligned square of the next larger volume. That is, a square of size
$s$ with $2^{i-1} < s < 2^i$ for some $i \in \{0, -1,-2,\ldots\}$ is
assigned to an $i$-pixel. This approach results in space that cannot
be used to assign squares, even though the remaining capacity
would suffice, and we can no longer guarantee to fit every valid sequence of
squares into the unit square. However,
we can guarantee to pack every such sequence into a $4$-underallocated
unit square (i.e., a $2 \times 2$ square), as every square is assigned
to a pixel that can hold no more than four times its volume. Most
importantly, our reallocation schemes continue to work in this setting
without any modifications. An example allocation is shown in
Figure~\ref{fig:quadtree}, where solid gray areas are assigned squares
and shaded areas indicate wasted space.

Note that a satisfactory reallocation scheme for arbitrary squares
with no or next to no underallocation is unlikely. Even the problem
of handling a sequence of insertions of total volume at most one,
without considering dynamic deletions and reallocation, requires
underallocation. This problem is known as {\em online square packing}, see
Fekete and Hoffmann~\cite{fh-ossp-13,fh-ossp-17}; 
currently, the best known approach results in
$5/2$-underallocation, see Brubach~\cite{brubach_improved_2014}.
}

\begin{figure}[h!]
\centering
\includegraphics[width=.7\linewidth]{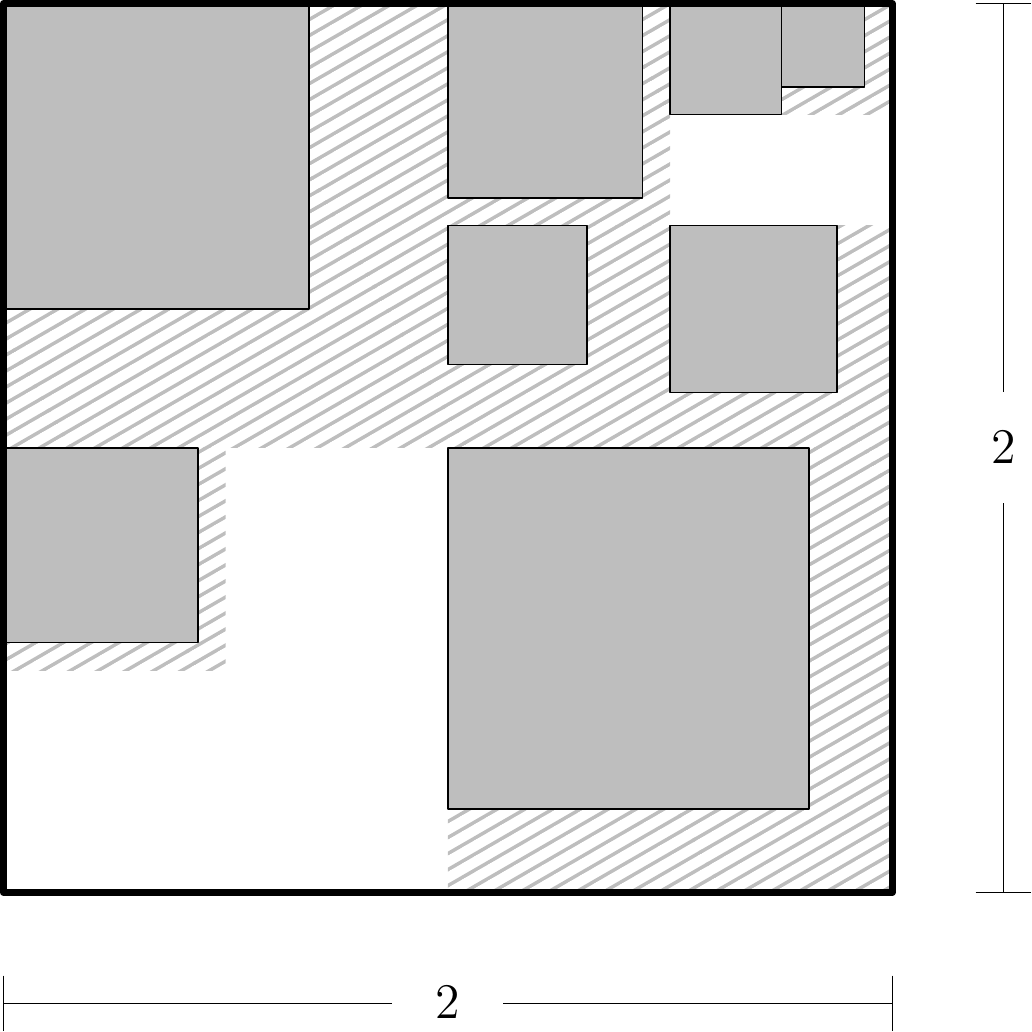}
\caption{\label{fig:quadtree}Example of a dynamically generated
  quadtree layout. {\color{black}The solid gray areas are packed
  squares. Shaded areas represent space lost due to rounding.}
}
\end{figure}

Rectangles of bounded aspect ratio $k$ are dealt with in the same way.
Also accounting for intermodule communication, every rectangle is
padded to the size of the next largest aligned square and assigned to
the node of a quadtree, at a cost not exceeding a factor of $4k$
compared to the one we established for the worst case.
{\color{black} As described in the following section, this theoretical bound
is rather pessimistic: the performance in basic simulation runs
is considerably better.}

\section{\color{black}Simulation Results}\label{sec:experiments}
We carried out a number of {\color{black} simulation runs to get an idea of the potential performance
of our approach}. For each test, we generated a random sequence of $1000$ requests 
that were chosen as \textsc{Insert($\cdot$)} (probability $0.7$) or \textsc{Delete($\cdot$)} (probability $0.3$). 
We apply a larger probability for \textsc{Insert($\cdot$)} to avoid the (relatively simple)
situation that repeatedly just a few rectangles are inserted and deleted, and in order
to observe the effects of increasing congestion. The individual modules were generated
by considering an upper
bound $b \in [0,1]$ for the side lengths of the considered squares. For
$b=0.125$, the value of the current underallocation seems to be stable except
for the range of the first $50$-$150$ requests. For $b=1$, the current
underallocation may be unstable, which could be caused by the following
simple observation: A larger $b$ allows larger rectangles that induce
$4k$-underallocations.

Our {\color{black} simulations indicate 
the theoretical worst-case bound of $1/4k$ may be overly pessimistic, see Figures~\ref{fig:experimentsA}--~\ref{fig:experimentsF}}. 
\textcolor{black}{In particular, the $x$-axis represents the number of operations and the $y$-axis represents the inverse value of underallocations. Furthermore, the red curves illustrate the inverse values of the underallocation and lie below the worst case values of $4k$.}
Taking into account that a purely one-dimensional approach cannot provide
an upper bound on the achievable underallocation, this {\color{black} provides reason to be optimistic about the potential
practical performance.}

{\color{black} A simulation of} the First-Fit approach for different values of
$k$ and upper bounds of $b = 0.125$ and $b=1$ for the side length of the
considered squares {\color{black} is shown in Figures~\ref{fig:experimentsA}--~\ref{fig:experimentsF}.
Each diagram illustrates the results of a simulation of $1000$ requests
that are randomly chosen as \textsc{Insert($\cdot$)} (probability $0.7$) or
\textsc{Delete($\cdot$)} (probability $0.3$).  We apply a larger probability
for \textsc{Insert($\cdot$)} to avoid the situation that repeatedly just a few
rectangles are inserted and deleted. The red graph shows the total current
underallocation after each request. The green graph shows the average of the
total underallocation in the range between the first and the current request.
We denote by~$c$ the number of collisions, i.e., the situations in that an
\textsc{Insert($\cdot$)} cannot be processed.}

\begin{figure}[ht]
\begin{center}
	\includegraphics[height=5.2cm]{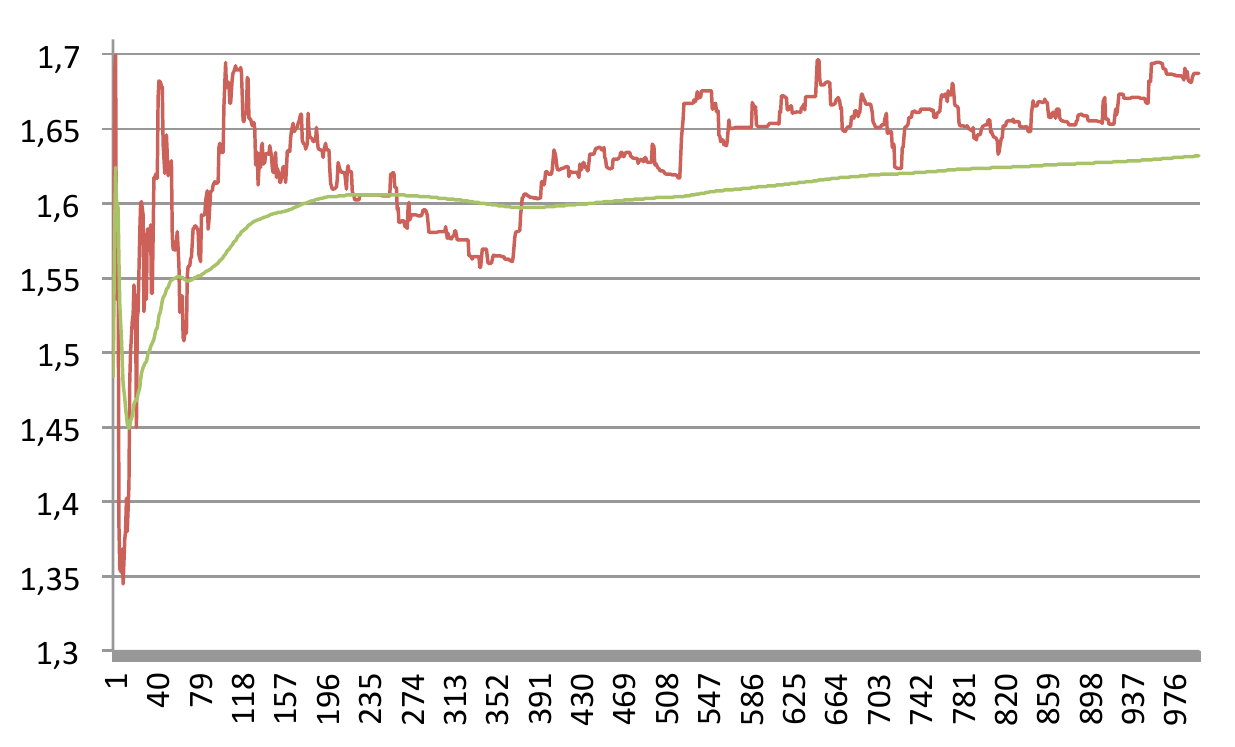}
 \end{center}
 \caption{\textcolor{black}{Number of operations ($x$-axis) vs. the inverse value of underallocation ($y$-axis) for the setting} $k=1$, $b=0.125$, $c=219$}
\label{fig:experimentsA}
\end{figure}

\begin{figure}[ht]
\begin{center}
	\includegraphics[height=5.2cm]{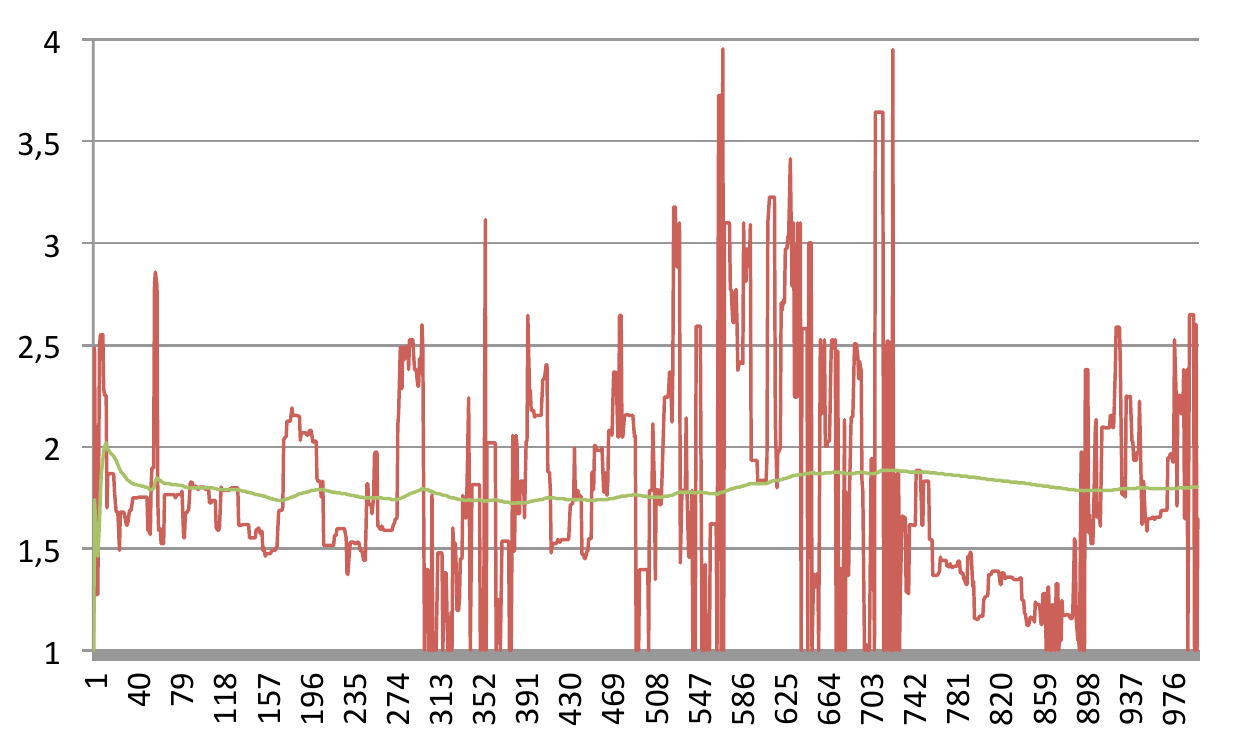}
 \end{center}
 \caption{\textcolor{black}{Number of operations ($x$-axis) vs. the inverse value of underallocation ($y$-axis) for the setting} $k=1$, $b=1$, $c=419$}
\label{fig:experimentsB}
\end{figure}

\begin{figure}[ht]
\begin{center}
	\includegraphics[height=5.2cm]{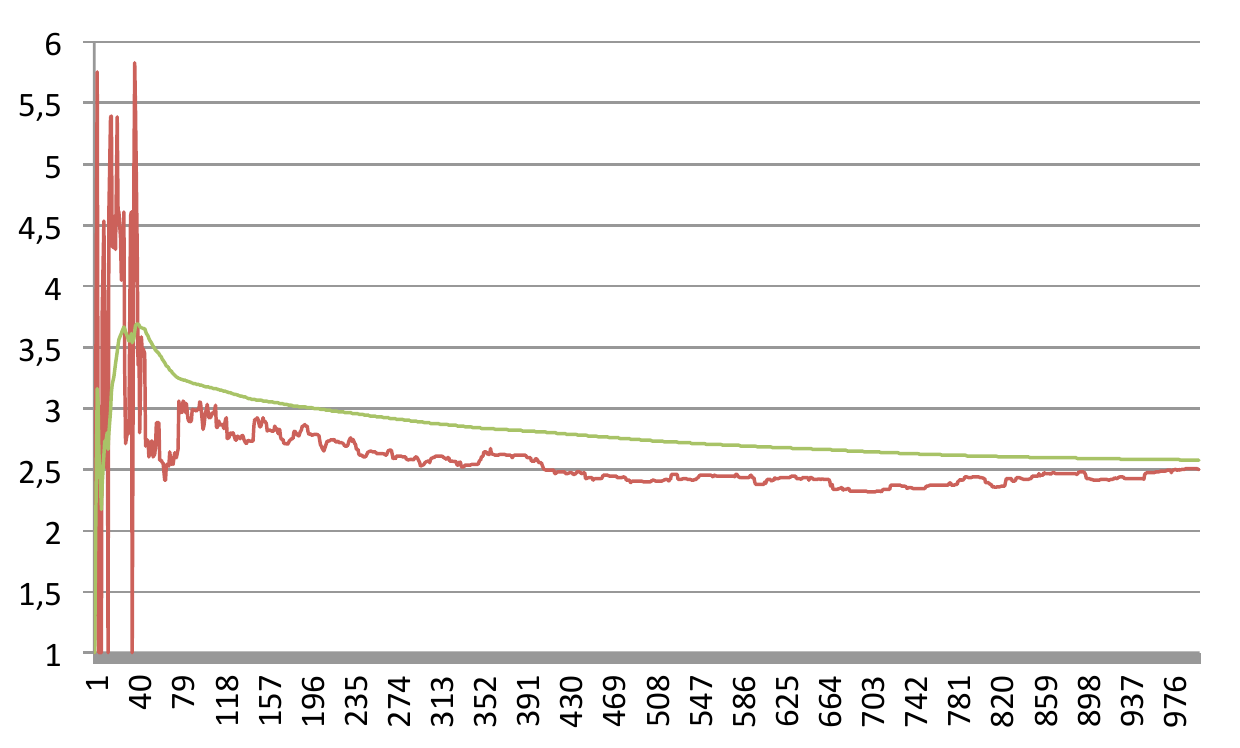}
 \end{center}
 \caption{\textcolor{black}{Number of operations ($x$-axis) vs. the inverse value of underallocation ($y$-axis) for the setting} $k=2$, $b=0.125$, $c=232$}
\label{fig:experimentsC}
\end{figure}

\begin{figure}[ht]
\begin{center}
	\includegraphics[height=5.2cm]{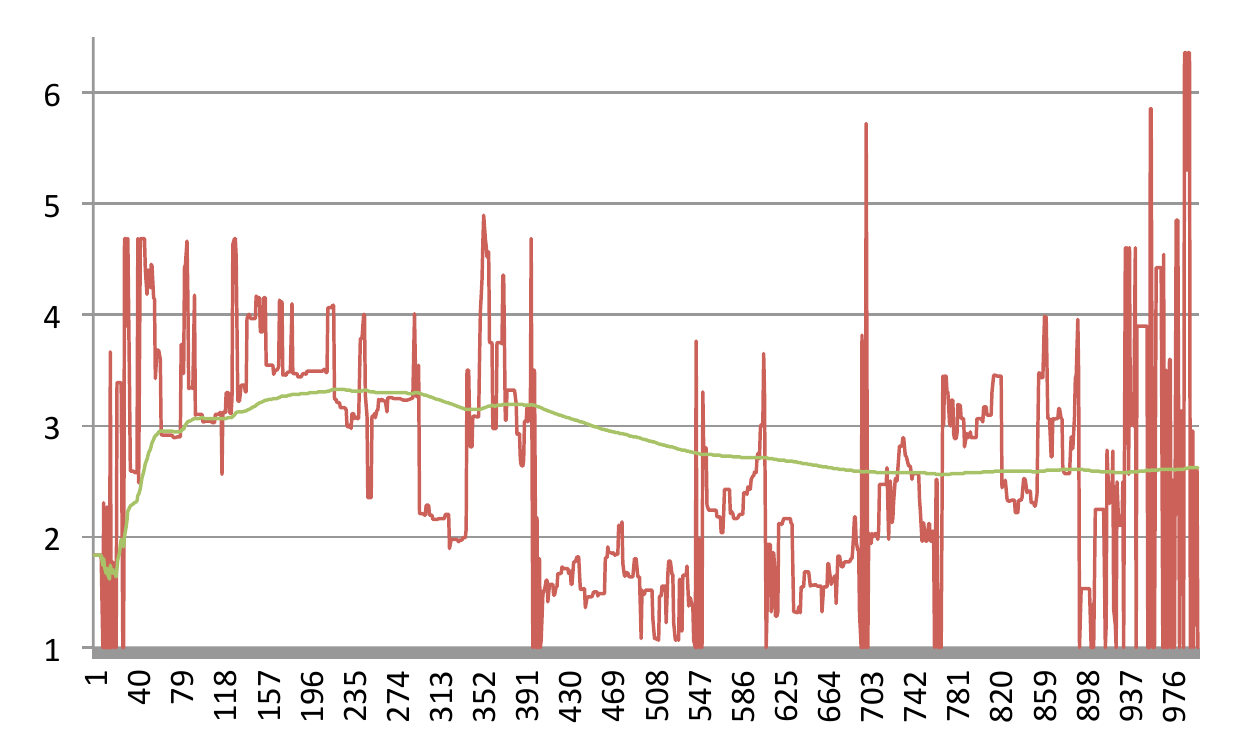}\\
 \end{center}
 \caption{\textcolor{black}{Number of operations ($x$-axis) vs. the inverse value of underallocation ($y$-axis) for the setting} $k=2$, $b=1$, $c=438$}
\label{fig:experimentsD}
\end{figure}

\begin{figure}[t]
\begin{center}
	\includegraphics[height=5.2cm]{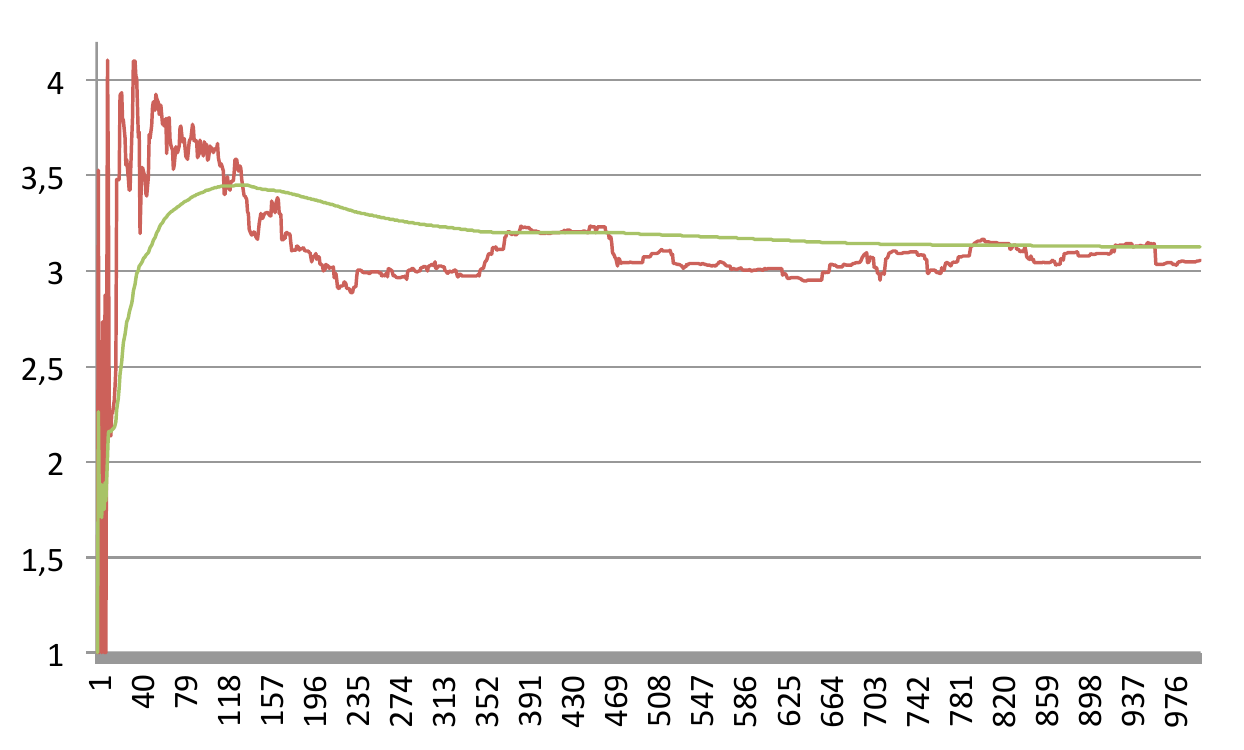}
 \end{center}
 \caption{\textcolor{black}{Number of operations ($x$-axis) vs. the inverse value of underallocation ($y$-axis) for the setting} $k=5$, $b=0.125$, $c=264$}
\label{fig:experimentsE}
\end{figure}

\begin{figure}[t]
\begin{center}
	\includegraphics[height=5.2cm]{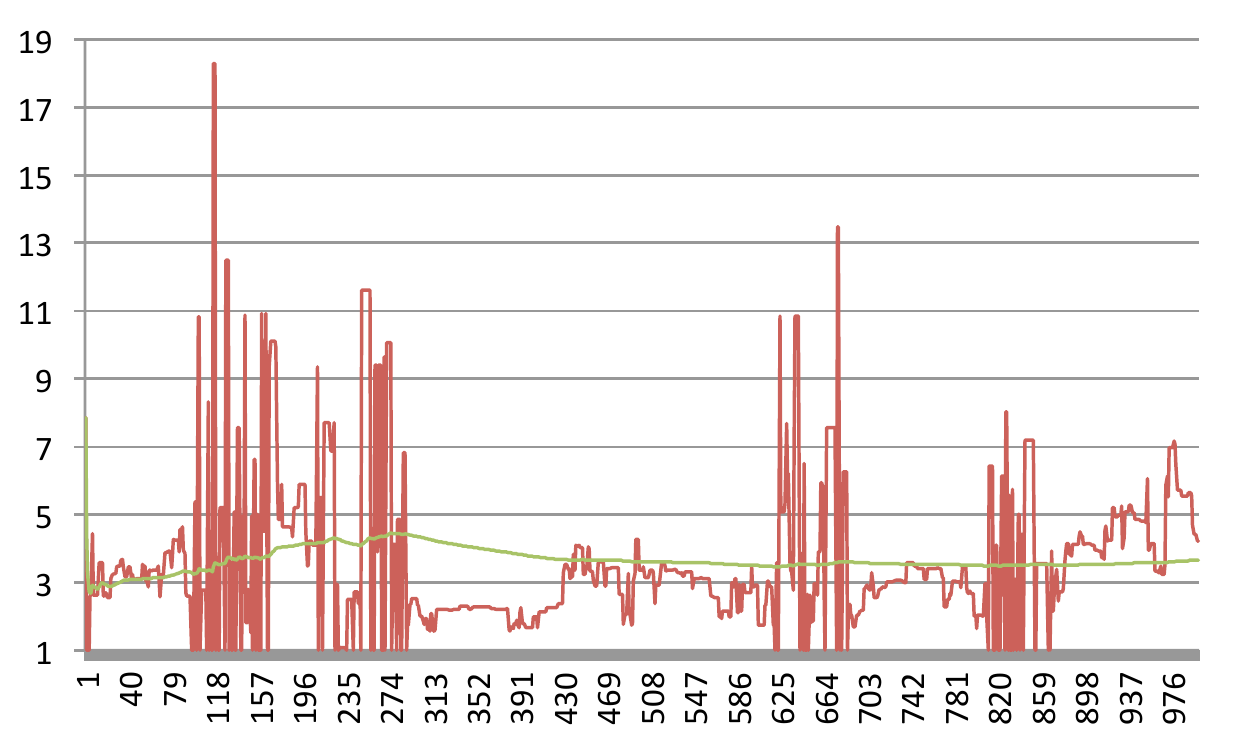}\\
 \end{center}
 \caption{\textcolor{black}{Number of operations ($x$-axis) vs. the inverse value of underallocation ($y$-axis) for the setting} $k=5$, $b=1$, $c=421$}
\label{fig:experimentsF}
\end{figure}

\section{Conclusions}
\label{sec:conc}

We have presented a data structure for exploiting
the full dimensionality of dynamic geometric storage and reallocation
tasks, such as online maintenance of the module layout for an FPGA.
These first results indicate that our approach is suitable for
making progress over purely one-dimensional approaches.
There are several possible refinements and extensions, including
a more sophisticated way of handling rectangles inside of 
square pieces of the subdivision, handling heterogeneous chip areas,
and advanced algorithmic methods. These will be addressed in future work.

{\color{black}
Another aspect of forthcoming work is an explicitly self-refining intermodule wiring. 
As indicated in Section~3 (and illustrated in Figure~\ref{fig:config}),
dynamically maintaining this communication infrastructure can be envisioned along the subdivision
of the recursive quadtree structure: making the routing a certain proportion of each cell area provides
a dynamically adjustable bandwidth, along with intersection-free routing, as shown in Figure~\ref{fig:config}.
First steps in this direction have been taken with an MA thesis~\cite{meyer}, with more work
to follow; this also addresses the aspect of robustness of communication in a hostile
environment that may cause individual connections to fail.}

\newpage
\balance
\bibliography{literature,mtjmr,lit}
\end{document}